\documentclass[leqno,oneside]{amsart}
\usepackage{babel}
\usepackage{amsmath}
\usepackage{amssymb}
\usepackage{amsthm}
\usepackage{geometry}
\usepackage{dsfont}
\usepackage{stmaryrd}
\usepackage{xcolor}
\usepackage{graphicx}

\usepackage[utf8]{inputenc}

\newtheorem{theorem}{Theorem}[section]
\newtheorem{lemma}[theorem]{Lemma}
\newtheorem{definition}[theorem]{Definition}
\newtheorem{question}[theorem]{Question}

\newtheorem{proposition}[theorem]{Proposition}

\newtheorem{remark}[theorem]{Remark}

\numberwithin{equation}{section}

\begin{document}

\title{On the lower bound for the length of minimal codes}

\author{Martin Scotti}
\address{Universit\'e Paris 8, Laboratoire de G\'eom\'etrie, Analyse et Applications, LAGA, Universit\'e Sorbonne Paris Nord, CNRS, UMR 7539, France.}
\email{martin.scotti@etud.univ-paris8.fr}
\thanks{M.~Scotti is supported by the ANR-21-CE39-0009 - BARRACUDA (French \emph{Agence Nationale de la Recherche}).}
\date{}
\maketitle

\begin{abstract}
In recent years, many connections have been made between minimal codes, a classical object in coding theory, and other remarkable structures in finite geometry and combinatorics. One of the main problems related to minimal codes is to give lower and upper bounds on the length $m(k,q)$ of the shortest minimal codes of a given dimension $k$ over the finite field $\mathbb{F}_q$. It has been recently proved that $m(k, q) \geq (q+1)(k-1)$. 

In this note, we prove that $\liminf_{k \rightarrow \infty} \frac{m(k, q)}{k} \geq (q+ \varepsilon(q) )$, where $\varepsilon$ is an increasing function such that $1.52 <\varepsilon(2)\leq \varepsilon(q) \leq \sqrt{2} + \frac{1}{2}$. Hence, the previously known lower bound is not tight for large enough $k$. We then focus on the binary case and prove some structural results on minimal codes of length $3(k-1)$. As a byproduct, we are able to show that, if $k = 5 \pmod 8$ and for other small values of $k$, the bound is not tight.
\end{abstract}

\section{Introduction}

In coding theory, minimal codes are an important centre of interest: in the binary case, they coincide with so-called intersecting codes, a classical and intensively studied object (see for example \cite{CZ, CL, Sloane}).
They also have interesting cryptographic properties: in particular, Massey showed an application of minimal codewords in secret sharing schemes (see \cite{M1, M2}). In \cite{AB}, the authors give a simple sufficient condition to get minimal codes, known as the Ashikhmin--Barg condition, upon which the investigation and construction of many minimal codes has been based over the past twenty years (see \cite{ding2015class,mesnager2020minimal,mesnager2019several} and references therein). 

More recently, an increasing number of connections with other areas of mathematics have been found, initiating a renewed interest in minimal codes.
For instance, in \cite{SP}, the authors implicitly use minimal codes to tackle the problem of the $j$-wise Davenport constant (a combinatorial invariant related to zero-sum sequences in finite abelian groups) in additive combinatorics. Another recently discovered link concerns the trifference problem (see 
the recently released preprint \cite{Bishnoi}). However, the most important connection unearthed is certainly the one with finite geometry.
A classical procedure in coding theory associates points in a projective space of dimension $k-1$ to a code of dimension $k$. In the case of minimal codes, as observed first in \cite{ABNgeo} and \cite{Tang} it turns out that this set of points in the projective space has the nice property of meeting every projective hyperplane in a subset spanning the hyperplane.
Such objects are called strong blocking sets, combinatorial structures introduced first in \cite{davydov2011linear} in connection with saturating sets. Note that a connection of these objects with minimal codes was already done in \cite{bonini2021minimal}, where they were named cutting blocking sets. 
Let us underline that the length of a minimal code of given dimension is the size of the corresponding strong blocking set.

While it is easy to construct strong blocking sets of large size (or equivalently, long minimal codes of a given dimension), it is not obvious how small they can be.
Therefore, most research focuses on determining bounds on the size of the shortest minimal code of a given dimension and to construct short minimal codes explicitly. The upper bound on the length of the shortest minimal code is a very active area of research. The current best known result has been obtained only recently, independently in \cite{ABN} and \cite{Bishnoi}: noting $m(k, q)$ the size of the shortest minimal code of dimension $k$ over $\mathbb{F}_q$, we have 
$$m(k,q) \leq \frac{2k}{\log_{q}\big(\frac{q^{4}}{q^{3} - q + 1}\big)}\cdot(q+1).$$
This upper bound is derived from non-effective existence results. Some explicit constructions of short minimal codes have been recently presented in \cite{ABN,3CB,BBconcat,Bartoli,Bishnoi,cohen2013minimal} and in the upcoming \cite{Expander}.

\medskip

In this note, we investigate the lower bound on the length of minimal codes.
So far the best known general result, proved in \cite{3CB}, is
$$m(k, q) \geq (q+1)(k-1).$$
In \cite{3CB}, along with the proof, the authors observe that this bound is not always tight.

Inspired by the usage of the MRRW bound in \cite{SP}, here we focus first on the asymptotic case when $q$ is fixed and $k$ is arbitrarily large, proving the following result.

\medskip

\noindent \textbf{Theorem A.} If $k$ is large enough, 
$$m(k, q) \geq \left(q+ \varepsilon(q) \right) \cdot k,$$
where $\varepsilon$ is an increasing function and $1.52 <\varepsilon(2)\leq \varepsilon(q) \leq \sqrt{2} + \frac{1}{2}$.

\medskip

Our method is purely coding theoretical, combining classical asymptotic bounds on codes and a bound on the minimum distance of minimal codes from \cite{3CB}. It does not seem to provide a geometric insight of the problem, that could give, instead, an even stronger result. During the writing process, it was brought to our attention that a similar result was proved independently in \cite{Bishnoi}, using a very similar argument. However, Theorem A is a slightly more explicit version of the results proved in \cite{Bishnoi}.

We then concentrate on the binary case, proving some structural results on minimal codes that meet the lower bound. In several dimensions it is possible to prove that there are no minimal codes satisfying these strong structural properties. Our main result is the following.

\medskip

\noindent \textbf{Theorem B.}  If $k \in \{7, 8, 9, 11\}$ or if $k = 5 \pmod 8$,
$$m(k, 2) \geq 3(k-1) + 1.$$

\medskip

\bigskip

\noindent \textbf{Outline:} Our note is organised as follows.
In Section 2, we introduce minimal codes and strong blocking sets, and give some of their properties that we shall use for the rest of the note.
In Section~3, we focus on the asymptotic case, using the well-known $q$-ary MRRW bound to deduce a new asymptotic lower bound for the length of minimal codes.
Finally, in Section 4, we tackle the binary case, showing some structural results on minimal codes of length $3(k-1)$ and using these to provide several dimensions $k$ where such short minimal codes cannot exist.

\section{Preliminaries}

\subsection{Linear codes} We denote by $\mathbb{F}_{q}$ the field with $q$ elements, where $q$ is a prime power.
A subspace of $\mathbb{F}_{q}^{n}$ of dimension $k$ is called a \emph{(linear) $[n, k]_{q}$-code}.
If $x$ is a vector of $\mathbb{F}_{q}^{n}$ we write $\sigma(x) = \{i \mid x_{i} \neq 0\}$ to denote its \emph{support}.
The space $\mathbb{F}_{q}^{n}$ is endowed with the \emph{Hamming metric} $wt_{H}(x) = |\sigma(x)|$.
The \emph{minimal distance} of a code is the smallest nonzero Hamming weight of a codeword.
Throughout this note we will consider $[n, k, d]$-codes, that is codes of dimension $k$ in $\mathbb{F}_{q}^{n}$ and with minimal distance $d$ (sometimes we will omit $d$ if it is not known).
Two $[n, k, d]_{q}$-codes $\mathcal{C}$ and $\mathcal{C}'$ are said to be equivalent if there is a linear map from $\mathbb{F}_{q}^{n}$ to itself that preserves Hamming weight  and maps $\mathcal{C}$ to $\mathcal{C}'$.

For an $[n, k, d]_{q}$-code, we define its \emph{rate} $R = \frac{k}{n}$ and its \emph{relative minimal distance} $\delta = \frac{d}{n}$.

\begin{definition}
A function $f : [0, 1] \rightarrow [0; 1]$ is \emph{asymptotically $q$-ary upper-bounding} if it is continuous, decreasing, and such that, for large enough length $n$,
$$R \leq f(\delta).$$
\end{definition}

For example, the well known Singleton bound $k+d \leq n+1$ produces the asymptotic upper-bounding function $f(x) = 1 - x$. In \cite{A}, Aaltonen generalises the celebrated asymptotic binary MRRW bound \cite{MRRW} to codes over any finite field. His generalised upper-bounding function for codes over $\mathbb{F}_{q}$ is
$$M(\delta) = H_{q}\left(\frac{1}{q}\left(q-1-(q-2)\delta - 2 \sqrt{(q-1)\delta(1-\delta)}\right)\right),$$
where $H_{q}$ is the $q$-ary entropy function
$$H_{q}(x) = -x\log_{q}\Big(\frac{x}{q-1}\Big) - (1-x)\log_{q}(1-x).$$

\subsection{Minimal codes and strong blocking sets}

Let us introduce the main object of this note.

\begin{definition}
In a code $\mathcal{C}$ a nonzero codeword $c \in \mathcal{C}$ is \emph{minimal} if there are no nonzero codewords $c' \in \mathcal{C}$ with $\sigma(c') \subsetneq \sigma(c)$.
A code $\mathcal{C}$ is \emph{minimal} if all of its nonzero codewords are minimal.
\end{definition}

An interesting property of minimal codes is that every codeword satisfies the following.

\begin{proposition}[{\cite[Proposition 1.5]{3CB}}] \label{prop:singleton}
Let $\mathcal{C}$ be a minimal code of parameters $[n, k, d]$.
Then, for every codeword $c \in \mathcal{C}$,
$$w_{H}(c) \leq n - k + 1$$
\end{proposition}

Recently, an important connection with finite geometry has been established, opening new scenarios of research.

The projective space of dimension $k-1$ is defined to be
$${\rm PG}(k-1, q) = \Big( \mathbb{F}_{q}^{k} \setminus \{ 0\} \Big) \big/ \sim$$
where $\sim$ is the equivalence relation defined by being colinear.

To a \emph{generator} matrix $G \in \mathcal{M}_{k, n}(\mathbb{F}_{q})$ of rank $k$ we associate the code $\mathcal{C}$ formed by its rowspan.
The code $\mathcal{C}$ is then an $[n, k]_{q}$-code.
The code $\mathcal{C}$ is said to be \emph{nondegenerate} if all columns of $G$ are nonzero.
In that case these columns (seen as vectors of $\mathbb{F}_{q}^{k}$) can be projected onto ${\rm PG}(k-1, q)$, where they form a multiset, i.e. a set with multiplicities.
If there all points have multiplicities one, the code is said to be \emph{projective}.

\begin{definition}
A \emph{strong blocking set} is a subset $B \subseteq {\rm PG}(k-1, q)$ such that for every hyperplane $H$ the intersection of $H$ and $B$ spans $H$. 
\end{definition}

\begin{theorem}[\cite{ABNgeo,Tang}]
A code is minimal if and only if the corresponding points in the projective space form a strong blocking set.
\end{theorem}

\subsection{Lower bound} By adding random columns to a generator matrix of a minimal code we get a minimal code. It is then interesting to know how short a minimal code can be.

\begin{definition}
The length of the shortest minimal codes over $\mathbb{F}_{q}$ of dimension $k$ is denoted $m(k, q)$.
\end{definition}

\begin{remark}
Since the length of a projective minimal code of dimension $k$ over $\mathbb{F}_{q}$ is also the size of the corresponding strong blocking set in ${\rm PG}(k-1, q)$, the function $m(k, q)$ also designates the size of the smallest strong blocking set in ${\rm PG}(k-1, q)$.
\end{remark}

\begin{theorem}[{\cite[Theorem 2.8]{3CB}}]\label{thm:minimumdistance}
Let $\mathcal{C}$ be a minimal code with parameters $[n, k, d]_{q}$. Then \[d \geq (q-1)(k-1) + 1.\]
\end{theorem}

\begin{theorem}[{\cite[Theorem 2.14]{3CB}}]\label{thm:lowerbound}
The length of a minimal codes of dimension $k$ over $\mathbb{F}_q$ is at least $(q+1)(k-1)$, that is
$$m(k, q) \geq (q+1)(k-1).$$
\end{theorem}

In Section $4$, we will frequently refer to the following geometric proof of Theorem \ref{thm:lowerbound}, which was given in \cite{HN}, that we report here for the reader's convenience. Before giving the proof, let us introduce some further objects.

\begin{definition}
A set $B \subseteq {\rm AG}(k-1, q)$ is called an \emph{affine blocking set} if it intersects every affine hyperplane of ${\rm AG}(k-1, q)$.
\end{definition}

\begin{theorem}[\cite{jamison1977covering}]\label{thm:Jamison}
Let $B \subseteq {\rm AG}(k-1, q)$ be an affine blocking set. Then $|B| \geq (q-1)(k-1) + 1$.
\end{theorem}

\begin{proof} [Proof of Theorem \ref{thm:lowerbound}]
Let $S$ be a strong blocking set of $PG(k-1, q)$.
Let $H$ be a hyperplane of ${\rm PG}(k-1, q)$ whose intersection with $S$ is maximal.
Write $S_{H} = S\setminus H$.
Since $S$ is a strong blocking set, $S_{H}$ is an affine blocking set.
Let $S_{H}'$ be a minimal affine blocking set (with respect to set inclusion) contained in $S_H$.
By Theorem \ref{thm:Jamison}, $|S_{H}'| \geq 1 + (k-1)(q-1)$.
By the minimality of $S_H'$, For every $P \in S_{H}'$ there is a hyperplane $U$ such that $U \cap S_{H}' = \{P\}$.
For any hyperplane $Z$ containing $H \cap U$ but different from $H$ and $U$ we have:
\begin{align*}
|H \cap S| \geq |Z \cap S| &\geq |Z \cap S_{H}| + |Z \cap S \cap H| = |Z\cap (S_{H} \setminus U)| + |U \cap (S \cap H)| \\
&= |Z \cap (S_{H} \setminus U)| + |S \cap H| + |(S \setminus H)\setminus U| - |S\setminus U| \\
&\geq (k-1) + |S\cap H| + (k-1)(q-1) - |S \setminus U|
\end{align*}
so that $|S \setminus U| \geq q(k-1)$. Since $|S \cap U| \geq (k-1)$, we get $|S| = |S \cap U| + |S \setminus U| \geq (k-1)(q+1)$.
\end{proof}

\section{The asymptotic lower bound}

In this section we prove an asymptotic improvement of the lower bound, using the $q$-ary MRRW bound from the last section and Theorem \ref{thm:minimumdistance}. Our proof generalises a result already known for $q=2$ (in \cite{CL} the authors credit Komlos with an unpublished proof but it seems that the first published version is \cite{KS}). More recently, a similar approach was used in \cite{SP} in the context of additive combinatorics.

\medskip

First, notice that Theorem \ref{thm:minimumdistance} implies that $g(x) = \frac{x}{q-1}$ is an asymptotic upper-bounding function for minimal codes over $\mathbb{F}_{q}$ (if we forget the condition that $g$ has to be decreasing).
Since $g$ is increasing, and since any asymptotic upper-bounding function for general $q$-ary codes is decreasing, it is possible to deduce an upper bound on the rate of minimal codes from any asymptotic upper-bounding function $f$, by simply computing the intersection between $f$ and $g$.

This method is shown in the following graph for $q=2$, using the MRRW bound as an upper-bounding function. The area in grey shows the region where minimal codes with large length (since we are considering asymptotic upper-bounding functions) may exist. The dotted line shows the maximum rate for a minimal code obtained with this method, which corresponds to a lower bound on the length when the dimension is fixed.

\begin{center}
\includegraphics[scale = 0.4]{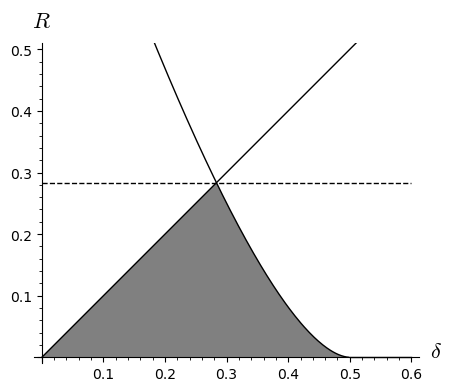}
\end{center}

According to the upper-bounding function used, we get different bounds. For example, the asymptotic Plotkin bound can be used to recover an asymptotic version of Theorem \ref{thm:lowerbound}.

\begin{lemma}\label{lem:plotkin}
For $k$ large enough we have
$$\frac{m(k, q)}{k} \geq q+1$$
\end{lemma}

\begin{proof}[Proof]
The asymptotic Plotkin bound states that $R(\delta) = 1 - \frac{q}{q-1}\cdot \delta$ is an asymptotic upper-bounding function. $R(\delta) = \frac{\delta}{q-1}$ implies $\delta = \frac{q-1}{q+1}$, that gives $R_{\max} = R(\delta) = \frac{1}{q+1}$.
\end{proof}

Our aim is to replace the Plotkin bound with Aaltonen's stronger $q$-ary MRRW bound, thus sharpening the asymptotic bound on strong blocking sets.

Some numeric evaluations of the intersection between $R$ and $g$ give the following lower bounds.

\begin{remark} \label{rem:table}
The following table gives lower bounds for the value of $\liminf_{k\rightarrow \infty} \frac{m(k, q)}{k}$ :

\begin{center}
\begin{table}[h!]
\begin{tabular}{c|c}
$q$ & $\liminf_{k \rightarrow \infty} \frac{m(k, q)}{k}$ \\ \hline
2   & 3.5276                                            \\
3   & 4.5516                                            \\
4   & 5.568                                             \\
5   & 6.5805                                            \\
7   & 8.5987                                            \\
8   & 9.6057                                           
\end{tabular}
\end{table}
\end{center}
\end{remark}
These lower bounds are all better than $q+1$, which would come from Theorem \ref{thm:lowerbound} or Lemma~\ref{lem:plotkin}. We aim to give an explicit statement of the asymptotic bound one can get with this method.

\begin{theorem} \label{thm:asymptoticbound}
For $k$ large enough we have
$$m(k, q) \geq (q+\varepsilon(q))\cdot k,$$
where $\varepsilon$ is an increasing function and verifies
$$1.5204 \leq \varepsilon(q) \leq \sqrt{2} + \frac{1}{2}.$$
\end{theorem}

\begin{proof}[Proof]
Recall the $q$-ary MRRW bound
$$M(\delta) = H_{q}\left(\frac{1}{q}\left(q-1-(q-2)\delta - 2 \sqrt{(q-1)\delta(1-\delta)}\right)\right).$$
We need to check that the intersection between $M$ and the upper bound on the rate of minimal codes obtained from Theorem \ref{thm:minimumdistance}, i.e. $R = \frac{\delta}{q-1}$, is smaller than $(q + \varepsilon(q))^{-1}$.
Since $M$ is decreasing, and since for $\delta = \delta_{c}(q) = \frac{q-1}{q+\varepsilon(q)}$ the linear upper bound on the rate is $R \leq \frac{\delta_{c}(q)}{q-1} = \frac{1}{q+\varepsilon(q)}$, it is enough to check that $M(\delta_{c}(q)) \leq \frac{1}{q+\varepsilon(q)}$, or equivalently, that
$$M\Big(\frac{q-1}{q+\varepsilon(q)}\Big) (q+ \varepsilon(q)) \leq 1.$$
First we write
$$A(q) = \frac{1}{q}\left(q-1-(q-2)\delta_{c}(q) - 2 \sqrt{(q-1)\delta_{c}(q)(1-\delta_{c}(q))}\right)
= \frac{q-1}{q(q+\varepsilon(q))}\cdot C(q),$$
where $C(q) = \varepsilon(q) + 2 - 2\sqrt{\varepsilon(q) + 1}$.\\
Computing $(q+\varepsilon(q))H_{q}(A(q))$, we get
$$(q+\varepsilon(q))H_{q}(A(q)) \leq \frac{q-1}{q}\cdot C(q)\cdot \log_{q}\Big(\frac{e}{C(q)} q(q+\varepsilon(q))\Big).$$
Recall that we want to prove $(q+\varepsilon(q))H_{q}(A(q)) \leq 1$, so that it is sufficient to establish
\begin{equation} \label{eq:ineq}
\frac{q-1}{q}\cdot C(q)\cdot \log_{q}\Big(\frac{e}{C(q)} q(q+\varepsilon(q))\Big) \leq 1.
\end{equation}
Let us now take $\varepsilon(q)$ so that \eqref{eq:ineq} is an equality. Numerical evaluations give $\varepsilon(2) \geq 1.5204$.

Suppose now that $\varepsilon$ has been shown to be increasing until the prime power $\ell$, and note $C_{\ell} = \varepsilon(\ell) + 2 - 2\sqrt{\varepsilon(\ell) + 1}$.
It is enough to show that
$$\forall q \geq \ell \quad \frac{q-1}{q}\cdot C_{\ell} \cdot \log_{q}\Big(\frac{e}{C_{\ell}}q(q+\varepsilon(\ell))\Big) \leq 1.$$
It is straightforward to check that the function defined by
$$f(x) = \frac{x-1}{x} \cdot C_{\ell} \cdot \frac{\ln\Big(\frac{e}{C_{\ell}}x(x+\varepsilon(\ell))\Big)}{\ln(x)}$$
is decreasing.
Since $f(\ell) = 1$, for all further primer powers $q$ larger than $\ell$, we have $f(q) < 1$, meaning that $\varepsilon(q) > \varepsilon(\ell)$.
This means that $\varepsilon$ is an increasing function.\\
Finally, direct computations yield
$$\lim_{q \rightarrow \infty} \varepsilon(q) = \sqrt{2} + \frac{1}{2}.$$
\end{proof}

\begin{remark}
Notice that because of the way that $\varepsilon$ is defined in the above proof, its values are suboptimal. If we had defined it to be the value such that $H_{q}(A(q))(q + \varepsilon(q)) = 1$, we would have gotten the bounds as they appear in Remark \ref{rem:table} but our proof that $\varepsilon$ is increasing would be invalid since $f(\ell) > 1$.
We have chosen the alternate definition precisely in order to show that $\varepsilon$ is increasing, judging that the gap is not too large.
We strongly believe that the actual values are also increasing.
Here is a table showing how far our expression for $\varepsilon$ is from the actual lower bound for small values of $q$.
\begin{center}
\begin{table}[h!]
\begin{tabular}{c|c|c}
$q$ & $\liminf_{k \rightarrow \infty} \frac{m(k, q)}{k} - q$ & $\varepsilon(q)$ \\ \hline
2   & 1.5276                                             & 1.5204           \\
3   & 1.5516                                             & 1.5450           \\
4   & 1.568                                              & 1.5624           \\
5   & 1.5805                                             & 1.5757           \\
7   & 1.5987                                             & 1.5951           \\
8   & 1.6057                                             & 1.6025
\end{tabular}
\end{table}
\end{center}
\end{remark}

\begin{remark}
Theorem \ref{thm:asymptoticbound} could be improved if some better $q$-ary asymptotic upper-bounding function was found.
Notice that it is sufficient for this $q$-ary asyptotic upper-bounding function to only concern minimal codes.
For instance, for $q=2$, if the Gilbert-Varshamov bound $f(x) = 1 - h_{2}(x)$ is indeed a binary asymptotic upper-bounding function, as many coding theorists seem to believe, then our method will yield a stronger asymptotic lower bound for the length of minimal codes as a result.
\end{remark}

\section{Short binary minimal codes}

In the previous section we have seen that $m(k, q) \geq (q+1)(k-1)$ is not tight when $k$ is large.
In \cite{3CB}, the authors already show that this bound is not tight when $2 \leq k \leq \sqrt{q} + 2$ (except when $q = 2$ and $k = 3$). In the present section we investigate this inequality when $k$ is small and $q=2$.

We will get different results, summarized in Theorem B in the introduction. We will first delve into the proof of Theorem \ref{thm:lowerbound}.
We will be able to deduce strong structure results for minimal codes for which equality in the bound holds.

\subsection{Some structure results}

Since we are looking for minimal-size strong blocking sets, many inequalities in the proof of 
Theorem \ref{thm:lowerbound} have to be equalities. Let $q=2$ and 
consider a strong blocking set $S \subseteq {\rm PG}(k-1, q)$ of size $(k-1)(q+1) = 3(k-1)$.
We must have $S_{H} = S_{H}'$, and both have to be of cardinality exactly $(k-1)(q-1) + 1 = k$.
We also must have $|S \cap U| = k-1$. In particular, $|S \cap U| = k-1$ must be true no matter what point $P$ and subsequent hyperplane $U$ we choose, meaning there are exactly $k$ such points, and therefore exactly $k$ hyperplanes $U$ for which $|S \cap U| = k-1$.

Recall that points in a strong blocking set correspond to columns of a generator matrix.

If we consider the $k$ points of $S_{H}$ to be the first $k$ columns, the codeword corresponding to the hyperplane $H$ is
$$c_{H} = (\underbrace{1,\dots,1}_{k \, \, \text{times}},\underbrace{0,\dots,0}_{2k-3 \, \, \text{times}}).$$

The codeword corresponding to a hyperplane $U_{j}$ (with $P$ being the point corresponding to the $j$-th column of the generator matrix) is
$$c_{U_{j}} = (1,\dots,1,0,1,\dots,1) \cdot b_{U_{j}}$$
where the $0$ is in $j$-th position and  $b_{U_{j}}$ is a codewords corresponding to the last $2k-3$ indices, of Hamming weight $w_{H}(b_{U_{j}}) = k-1$.

Since the codewords $c_{H}+c_{U_{j}}$ are in echelon form, they are linearly independent and, since there are exactly $k$ such codewords, they form the lines of a generator matrix
\begin{equation} \label{eq:matrix}
G = \begin{pmatrix}
I_{k} & P\\
\end{pmatrix}
\end{equation}

Each row of $P$ must have Hamming weight $k-1$, and since $c_{H}$ is a codeword that can only be generated by summing all the rows, all the columns of $P$ must have even Hamming weight.

We have proved then the following result.

\begin{proposition} \label{prop:matrix}
Let $\mathcal{C}$ be a binary minimal code of parameters $[n, k, d]$ with $n = 3(k-1)$.
Then $\mathcal{C}$ is equivalent to a code that has a generator matrix of form \eqref{eq:matrix}, where each row is a codeword of Hamming weight $k$ and each column of $P$ is of even Hamming weight.
Furthermore $d = k$.
\end{proposition}

Our purpose is to investigate wether it is possible for matrices of the form \eqref{eq:matrix} to be generator matrices of binary minimal codes. In order to simplify notations we will write $N = k-1$ throughout this part.

Consider a matrix of form \eqref{eq:matrix}.
The submatrix $P$ has $N+1$ rows and $2N-1$ columns.
We write $P_{j} \subseteq \{1, \dots, 2N-1\}$ the support of the $j$-th row.
Since each row of $P$ has Hamming weight $N$, each corresponding subset $P_{j}$ of $\{0, \dots, 2N-1\}$ has cardinality $N$.

The minimal code corresponding to this generator matrix has minimal distance $N+1$ (because of the maximality of the hyperplane $H$ defined above as the hyperplane with the largest intersection with the strong blocking set), and since the weight of each codeword satisfies the bound of Pro\-po\-si\-tion~\ref{prop:singleton}, the maximal allowed weight is $2N$.

\begin{lemma}\label{lem:corr}
If there is a $[3N,N+1]$ minimal code, then there exists a family of $N+1$ subsets of $\{1, \dots, 2N-1\}$ of size $N$ and pairwise intersection between $\frac{N-1}{2}$ and $\frac{N+1}{2}$, such that the symmetric difference of all $N+1$ subsets is the empty set.
\end{lemma}

\begin{proof}
We write $A\Delta B = (A \cup B) \setminus (A \cap B)$ the symmetric difference of the sets $A$ and $B$.
Adding the $i$-th and the $j$-th row of $G$ produces a codeword with weight $2 + |P_{i} \Delta P_{j}|$ (the term $2$ comes from the $I_{k}$ in $G$).
Adding the codeword $c_{H}$ to this sum produces a codeword with weight $N-1 + |P_{i} \Delta P_{j}|$.
Since both of these codewords have weights between $N+1$ and $2N$, we can deduce bounds on the size of $|P_{i} \Delta P_{j}|$: $N +1 \leq 2 + |P_{i} \Delta P_{j}|$ yields $|P_{i} \Delta P_{j}| \geq N-1$, while $N-1 + |P_{i} \Delta P_{j}| \leq 2N$ yields $|P_{i} \Delta P_{j}| \leq N+1$, so we get $N-1 \leq |P_{i} \Delta P_{j}| \leq N+1$.
Since $|P_{i} \Delta P_{j}| = 2N - 2|P_{i} \cap P_{j}|$ we get
\begin{equation}\label{eq:inter}
\frac{N-1}{2} \leq |P_{i} \cap P_{j}| \leq \frac{N+1}{2}
\end{equation}
Finally, since $c_{H}$ is a codeword and can only be generated by summing all rows of the generator matrix, all final $2N-1$ columns must be of even weight, i.e. the symmetric difference of the subsets $P_{i}$ must be empty.
\end{proof}

In particular, when $N = k-1$ is even, there is only one possible value for the cardinality of the intersection of the subsets $P_{i}$ and $P_{j}$.

Note that reasoning in terms of weights gives additional conditions for the size of the symmetric difference of the subsets corresponding to $3$ rows of $P$, and so on.

\subsection{The case when $N$ is even}

When $N$ is even, searching for minimal-size strong blocking sets it is possible to take $P_{1} = \{1, \dots, N\}$ and $P_{2} = \{N/2 +1, \dots, N + N/2 \}$ without loss of generality because of the above discussion.

For any further subset $I \subseteq \{1, \dots, 2N-1\}$ we write
\begin{align*}
a_{I} &= |I \cap (P_{1} \setminus P_{2}) | \\
b_{I} &= |I \cap (P_{1} \cap P_{2}) | \\
c_{I} &= |I \cap (P_{2} \setminus P_{1}) | \\
d_{I} &= |I \cap \big(\{1, \dots, 2N-1\}\setminus (P_{1} \cap P_{2})\big) | \\
\end{align*}
The following lemma presents some restrictions on the values that $a_{I},\dots, d_{I}$ can take.

\begin{lemma} \label{lem:structure}
Let $\mathcal{C}$ be a minimal binary code of dimension $k \geq 3$, and with an equivalent code that has a generator matrix of form \eqref{eq:matrix}.
If $P_{1}$, $P_{2}$, and $I$ are the supports of $3$ rows of the submatrix $P$ then the $a_{I}, \dots, d_{I}$ defined above verify $a_{I} = c_{I}$, $b_{I} = d_{I}$ and $|a_{I} - b_{I}| \leq 1$.
\end{lemma}

\begin{proof}
First, note that $a_{I} + b_{I} = |P_{1}\cap I| = N/2 = |P_{2}\cap I| = b_{I} + c_{I}$, yielding $a_{I} = c_{I}$. Furthermore, since $a_{I} + b_{I} + c_{I} + d_{I} = N = a_{I} + 2b_{I} + c_{I}$, we also have $b_{I} = d_{I}$.

Now note $s = b_{I} + d_{I}$.
The symmetric difference of all three sets has cardinality $N + s - (N - s) = 2s$.
This sum of three rows produces a codeword $c = c_{P_{1} + P_{2} + I}$ with weight $w_{H}(c) = 3 + 2s$.
Since $w_{H}(c) = 3 + 2s \geq N+1$ because of the weight conditions we get $s \geq N/2 - 1$.
Furthermore, remember that $c + c_{H}$ is also a codeword, of weight $N - 2 + 2s = N + 2s - 2$.
Now the other weight inequality yields $N + 2s - 2 \leq 2N$, giving $s \leq N/2 + 1$.
%\Martin{Many thanks to Wolfgang Schmid for finding an arithmetic error in this proof, which terrified me at first, but fortunately everything is all right}
This means that $N/2 - 1 \leq b_{I} + d_{I} \leq N/2 + 1$, which also implies that $N/2 - 1 \leq a_{I} + c_{I} \leq N/2 + 1$.
Finally, since $a_{I} = c_{I}$ and $b_{I} = d_{I}$ we get $|a_{I} - N/4| \leq 1/2$ and $|b_{I} - N/4| \leq 1/2$, giving $|a_{I} - b_{I}| \leq 1$.
\end{proof}

This lemma is particularly useful, as it allows us to restrict significantly which further subsets we may choose when we have already chosen the initial two.
Note that any new subsets must verify the above conditions for any two choices of $P_{1}$ and $P_{2}$ among the already chosen subsets, meaning that as $N$ increases we expect this to be more and more unlikely.

\medskip

\begin{proposition}
There is no binary $[18, 7]$ minimal code.
\end{proposition}

\begin{proof}
Here we have $P_{1} = \{1, 2, 3, 4, 5, 6\}$ and $P_{2} = \{4, 5, 6, 7, 8, 9\}$.
We are interested in the columns corresponding to $10$ and $11$, both of these must be of even weight.
Because of Lemma \ref{lem:structure}, the only possible values for $d_{I}$ are $1$ and $2$. We divide all possible subsets that satisfy Lemma \ref{lem:structure} into 3 families, depending on their intersection with $\{10, 11\}$:
$$A = \{ I \subseteq \{1, \dots, 2N-1\} \mid I \cap \{10, 11\} = \{10\} \}$$
$$B = \{ I \subseteq \{1, \dots, 2N-1\} \mid I \cap \{10, 11\} = \{11\} \}$$
$$C = \{ I \subseteq \{1, \dots, 2N-1\} \mid I \cap \{10, 11\} = \{10, 11\} \}$$
Now in order to have all rows sum to $0$ we must have $3$ subsets from one family and $1$ from both others. First suppose there are $3$ subsets from $\varepsilon(q)$ (by symmetry this also covers the case where there are $3$ from $\beta$). Then their symmetric difference is $\{4, 5, 6, 10\}$ and the symmetric difference of all five chosen subsets is $S = \{1, 2, 3, 4, 5, 6, 7, 8, 9, 10\}$, which produces a codeword of weight $5 + |S| = 5 + 10 = 15$ while the maximal allowed weight is $2N = 12$ by Proposition~\ref{prop:singleton}, a contradiction.
This means that there are $3$ subsets from $\gamma$.
In this case it can be checked that their symmetric difference has to be $\{1, 2, 3, 7, 8, 9, 10, 11\}$ and so the symmetric difference of all five codewords is $S' = \{10, 11\}$, which when added with $c_{H}$ produces a codeword of weight $(7-5) + |S'| = 2 + 2 = 4$, while the minimum allowed weight is $N+1 = 7$ (since the minimum distance has to be $d = 7$ by Lemma \ref{lem:structure}), again a contradiction.
\end{proof}

\begin{proposition}
There is no binary $[24,9]$ minimal code.
\end{proposition}

\begin{proof}
Here, after taking $P_{1}$ and $P_{2}$, Lemma \ref{lem:structure} forces all remaining subsets $I$ to have $d_{I} = 2$. Since $d_{I}$ is calculated by considering the intersection of a subset with $\{13, 14, 15\}$, and since there are $7$ subsets left, they form a $7\times3$ matrix where each column has an even amount of $1$'s, and so there is at least one column with at least three $0$'s.
Translated back into subsets, this means that there have to be $3$ subsets containing (without loss of generality) $\{13, 14\}$.
Without loss of generality assume that the first one is $I$.
Since $a_{I} = b_{I} = c_{I} = d_{I} = 2$ without loss of generality we choose $I = \{1, 2, 5, 6, 9, 10, 13, 14\}$. Now applying Lemma \ref{lem:structure} with all combinations of two already chosen subsets, we obtain that the only subset left that contains $\{13, 14\}$ is $J = \{3, 4, 5, 6, 11, 12, 13, 14\}$, while we need at least two in order to find three subsets containing $\{13, 14\}$ as required, a contradiction.
\end{proof}

\begin{proposition}
There is no binary $[30,11]$ minimal code.
\end{proposition}

\begin{proof}
This was proved by a computer search with the programming language Sagemath in $248$ seconds on a single 9th gen. Intel i5 core processor, restricting all remaining subsets of $\{1, \dots, 2N-1\}$ as more and more are selected.
This quickly yields the impossibility of any $11$ sets satisfying all conditions, and thus there is no strong blocking set of size $30$.
\end{proof}

\begin{theorem}
If $N = 4 \pmod 8$, there is no binary $[3N,N+1]$ minimal code.
\end{theorem}

\begin{proof}
As a direct consequence of the proof of Lemma \ref{lem:structure}, when $N = 0 \pmod 4$, we have $a_{I} = b_{I} = c_{I} = d_{I} = N/4$ for any new subset $I$. Consider the submatrix of size $N-1 \times N/4$ corresponding to the indices $\{3N/2+1, \dots, 2N-1\}$ (i.e. the indices determining the $d_{I}$'s). All of its rows are the indicator functions of the $N-1$ remaining subsets on $\{3N/2+1, \dots, 2N-1\}$, and so each line of this submatrix has exactly $N/4$ $1$'s. When $N = 4 \pmod 8$, we get that $N/4$ is odd.
Since $N$ is even, $N-1$ is odd too. This means that the whole submatrix has an odd amount of $1$'s, and thus that its rows cannot sum to $0$, contradicting Proposition \ref{prop:matrix}.
\end{proof}

From what we said above and by the asymptotic results of the previous section, it is natural to ask the following.

\begin{question}
Is there a binary minimal code of parameters $[3N, N+1]$ if $N \geq 4$ is even?
\end{question}

We expect a negative answer, but we have no general arguments for the missing cases.

\subsection{The case when $N$ is odd}

When $N$ is odd the situation is significantly more complicated, as two subsets of size $N$ satisfying \eqref{eq:inter} have an intersection whose cardinality is not reduced to a single integer, in fact two are possible.
This means that once we select a subset of size $N$ from $\{1, \dots, 2N-1\}$, say $\{1, \dots, N\}$ (without loss of generality), eliminating subsets of size $N$ as candidates for other rows becomes significantly harder. This makes it more likely for a family of $N+1$ subsets of size $N$ of $\{1, \dots, 2N - 1\}$ to satisfy Lemma \ref{lem:corr}.
Indeed, when $k$ is even (and thus $N$ odd) there are examples of minimal codes of length $3(k-1)$ when $k \in \{2, 4, 6\}$ (for instance in \cite{BBconcat}), while the above subsection shows that this phenomenon is rarer when $k$ is odd.
The following result shows that this streak also interrupts quite early.

\begin{proposition}
When $N = 7$ there is no strong blocking set of size $3N = 21$.
\end{proposition}

\begin{proof}
This was done via a computer search using the programming language Sagemath in $802$ seconds on a single 9th gen. Intel i5 core processor, using a similar method to the one for $N = 10$.
\end{proof}

\section{Acknowledgements}

The author is grateful to Martino Borello and Wolfgang Schmid for their review of this work and their insightful advice.

\bibliography{articles}

\begin{thebibliography}{10}

\bibitem{A}
M.~Aaltonen.
\newblock Linear programming bounds for tree codes (corresp.).
\newblock {\em IEEE Transactions on Information Theory}, 25(1):85--90, 1979.

\bibitem{ABNgeo}
G.~N. Alfarano, M.~Borello, and A.~Neri.
\newblock A geometric characterization of minimal codes and their asymptotic
  performance.
\newblock {\em Advances in Mathematics of Communications}, 16(1):115--133,
  2022.

\bibitem{ABN}
G.~N. Alfarano, M.~Borello, and A.~Neri.
\newblock Outer strong blocking sets.
\newblock {\em arXiv preprint 2301.09590}, 2023.

\bibitem{3CB}
G.~N. Alfarano, M.~Borello, A.~Neri, and A.~Ravagnani.
\newblock Three combinatorial perspectives on minimal codes.
\newblock {\em SIAM Journal on Discrete Mathematics}, 36(1):461--489, 2022.

\bibitem{Expander}
N.~Alon, S.~Das, and A.~Neri.
\newblock Strong blocking sets and minimal codes from expander graphs.
\newblock {\em preprint}, 2023.

\bibitem{AB}
A.~E. Ashikhmin and A.~Barg.
\newblock Minimal vectors in linear codes.
\newblock {\em IEEE Transactions on Information Theory}, 44:2010--2017, 1998.

\bibitem{BBconcat}
D.~Bartoli and M.~Borello.
\newblock Small strong blocking sets by concatenation.
\newblock {\em SIAM Journal on Discrete Mathematics}, 37(1):65--82, 2023.

\bibitem{Bartoli}
D.~Bartoli, A.~Cossidente, G.~Marino, and F.~Pavese.
\newblock On cutting blocking sets and their codes.
\newblock {\em Forum Mathematicum}, 34:347 -- 368, 2020.

\bibitem{Bishnoi}
A.~Bishnoi, J.~D'haeseleer, D.~Gijswijt, and A.~Potukuchi.
\newblock Blocking sets, minimal codes and trifferent codes.
\newblock {\em arXiv preprint 2301.09457}, 2023.

\bibitem{bonini2021minimal}
M.~Bonini and M.~Borello.
\newblock Minimal linear codes arising from blocking sets.
\newblock {\em Journal of Algebraic Combinatorics}, 53:327--341, 2021.

\bibitem{CZ}
G.~Cohen and G.~Zemor.
\newblock Intersecting codes and independent families.
\newblock {\em IEEE Transactions on Information Theory}, 40(6):1872--1881,
  1994.

\bibitem{CL}
G.~D. Cohen and A.~Lempel.
\newblock Linear intersecting codes.
\newblock {\em Discrete Mathematics}, 56:35--43, 1984.

\bibitem{cohen2013minimal}
G.~D. Cohen, S.~Mesnager, and A.~Patey.
\newblock On minimal and quasi-minimal linear codes.
\newblock In {\em Cryptography and Coding: 14th IMA International Conference,
  IMACC 2013, Oxford, UK, December 17-19, 2013. Proceedings 14}, pages 85--98.
  Springer, 2013.

\bibitem{davydov2011linear}
A.~A. Davydov, M.~Giulietti, S.~Marcugini, and F.~Pambianco.
\newblock Linear nonbinary covering codes and saturating sets in projective
  spaces.
\newblock {\em Advances in Mathematics of Communications}, 5(1):119, 2011.

\bibitem{ding2015class}
K.~Ding and C.~Ding.
\newblock A class of two-weight and three-weight codes and their applications
  in secret sharing.
\newblock {\em IEEE Transactions on Information Theory}, 61(11):5835--5842,
  2015.

\bibitem{HN}
T.~Heger and Z.~L. Nagy.
\newblock Short minimal codes and covering codes via strong blocking sets in
  projective spaces.
\newblock {\em IEEE Transactions on Information Theory}, 68:881--890, 2021.

\bibitem{jamison1977covering}
R.~E. Jamison.
\newblock Covering finite fields with cosets of subspaces.
\newblock {\em Journal of Combinatorial Theory, Series A}, 22(3):253--266,
  1977.

\bibitem{KS}
G.~Katona and J.~Srivastava.
\newblock Minimal 2-coverings of a finite affine space based on gf(2).
\newblock {\em Journal of Statistical Planning and Inference}, 8(3):375--388,
  1983.

\bibitem{M1}
J.~L. Massey.
\newblock Minimal codewords and secret sharing.
\newblock pages 276--279, 1993.

\bibitem{M2}
J.~L. Massey.
\newblock Some applications of coding theory in cryptography.
\newblock {\em Codes and Ciphers: Cryptography and Coding IV}, pages 33--47,
  1995.

\bibitem{MRRW}
R.~McEliece, E.~Rodemich, H.~Rumsey, and L.~Welch.
\newblock New upper bounds on the rate of a code via the delsarte-macwilliams
  inequalities.
\newblock {\em IEEE Transactions on Information Theory}, 23(2):157--166, 1977.

\bibitem{mesnager2020minimal}
S.~Mesnager, Y.~Qi, H.~Ru, and C.~Tang.
\newblock Minimal linear codes from characteristic functions.
\newblock {\em IEEE Transactions on Information Theory}, 66(9):5404--5413,
  2020.

\bibitem{mesnager2019several}
S.~Mesnager and A.~S{\i}nak.
\newblock Several classes of minimal linear codes with few weights from weakly
  regular plateaued functions.
\newblock {\em IEEE Transactions on Information Theory}, 66(4):2296--2310,
  2019.

\bibitem{SP}
A.~Plagne and W.~A. Schmid.
\newblock An application of coding theory to estimating davenport constants.
\newblock {\em Designs, Codes and Cryptography}, 61:105--118, 2010.

\bibitem{Sloane}
N.~J. Sloane.
\newblock Covering arrays and intersecting codes.
\newblock {\em Journal of combinatorial designs}, 1(1):51--63, 1993.

\bibitem{Tang}
C.~Tang, Y.~Qiu, Q.~Liao, and Z.~Zhou.
\newblock Full characterization of minimal linear codes as cutting blocking
  sets.
\newblock {\em IEEE Transactions on Information Theory}, 67(6):3690--3700,
  2021.

\end{thebibliography}
\bibliographystyle{abbrv}

\end{document}